\documentclass[11pt,letterpaper]{article}
\usepackage{cite}
\usepackage{amsmath,amssymb,amsfonts}
\usepackage{graphicx}
\usepackage{textcomp}
\usepackage{xcolor}
\usepackage{amsthm}
\usepackage[margin=1in]{geometry}
\usepackage{authblk}

\usepackage{amsmath,amssymb}
\usepackage{bm}
\usepackage{graphicx}
\usepackage{ascmac}
\usepackage{cite}
\usepackage{amsmath,amssymb, amsthm}
    \usepackage{bm}
    \usepackage{graphicx}
    \usepackage{ascmac}
    \usepackage{color}
    \usepackage{algorithm}		
    \usepackage[noend]{algorithmic}
    \usepackage{setspace}
    \usepackage{comment}
    \usepackage{footnote}

    \usepackage{color}

    \newcommand{\expect}[1]{\mathbb{E}[{#1}]}

\newtheorem{definition}{Definition}
\newtheorem{proposition}{Proposition}
\newtheorem{lemma}{Lemma}
\newtheorem{corollary}{Corollary}
\newtheorem{theorem}{Theorem}

\begin{document}

\title{Fast Neighborhood Rendezvous }

\author{Ryota Eguchi\footnote{Nagoya Institute of Technology, Nagoya, Aichi, Japan. E-mail: 30514002@stn.ntiech.ac.jp} 
\and Naoki Kitamura\footnote{Nagoya Institute of Technology, Nagoya, Aichi, Japan. E-mail: ktmr522@yahoo.co.jp.} 
\and Taisuke Izumi\footnote{Osaka University, Suita, Osaka, Japan. E-mail: t-izumi@ist.osaka-u.ac.jp.}}

\date{}

\maketitle

\begin{abstract} 
    In the rendezvous problem, two computing entities (called \emph{agents}) located at different vertices in a graph have to meet at the same vertex. In this paper, we consider the synchronous \emph{neighborhood rendezvous problem},
    where the agents are initially located at two adjacent vertices. While
    this problem can be trivially solved in $O(\Delta)$ rounds ($\Delta$
    is the maximum degree of the graph), it is highly challenging to 
    reveal whether that problem can be solved in $o(\Delta)$ rounds, even 
    assuming the rich computational capability of agents. The only known
    result is that the time complexity of $O(\sqrt{n})$ rounds is achievable
    if the graph is complete and agents are probabilistic, asymmetric, and can use 
    whiteboards placed at vertices. Our main contribution is to clarify the situation (with respect to 
    computational models and graph classes) admitting such a sublinear-time rendezvous algorithm. 
    More precisely, we present two algorithms achieving fast rendezvous additionally assuming 
    bounded minimum degree, unique vertex identifier, accessibility to neighborhood IDs, and randomization.
    The first algorithm runs within $\tilde{O}(\sqrt{n\Delta/\delta} + n/\delta)$ rounds for graphs of the minimum degree larger than $\sqrt{n}$, where $n$ is the number of vertices in the graph, and $\delta$ is the minimum degree of the graph. The second algorithm assumes that 
    the largest vertex ID is $O(n)$, and achieves $\tilde{O}\left( \frac{n}{\sqrt{\delta}} \right)$-round time complexity without using whiteboards. These algorithms attain
    $o(\Delta)$-round complexity in the case of $\delta = {\omega}(\sqrt{n} \log n)$ and $\delta = \omega(n^{2/3} \log^{4/3} n)$ 
    respectively. We also prove that four unconventional assumptions of our algorithm, 
    bounded minimum degree, accessibility to neighborhood IDs, initial distance one, and randomization are all 
    inherently necessary for attaining fast rendezvous. That is, one can obtain the $\Omega(n)$-round lower bound 
    if either one of them is removed.
    \end{abstract}
    
\footnote[0]{A preliminary subset of this work appeared in the
proceedings of ICDCS 2020~\cite{preliminaryversion}}

\section{Introduction}
\label{sec:introduction}

\subsection{Background}
The \emph{rendezvous problem} is well-studied in distributed computing theory. 
A typical setting of the problem requires two agents located at any vertices 
in a graph $G=(V, E)$ to meet and halt. This is recognized as a 
fundamental problem for designing distributed algorithms of mobile agents.
The hardness of symmetry breaking is often seen as an essential difficulty of the rendezvous problem.
For example, we consider a ring network of four vertices, and the situation that the two agents located at two vertices that are not adjacent to each 
other\footnote{More precisely, each port number assigned to each edge is also symmetric (for example, edges of clockwise direction have port number one, and edges of counter-clockwise direction have port number zero.).}. Then, the agents running the same algorithm symmetrically move and thus their relative distance two is kept forever. That is, any deterministic algorithm does not achieve rendezvous in this situation. To make the rendezvous problem solvable, the system model must be equipped with some mechanism enabling two agents to move asymmetrically.  Much of the previous work focuses on what models or assumptions provide such a capability~\cite{kranakis2003mobile,flocchini2004multiple,de2006asynchronous}.

Unlike the viewpoint mentioned above, we assume a model that easily breaks symmetry, i.e., 
allowing randomized and/or asymmetric algorithms, and focuses on the time complexity of the rendezvous 
problem. 
When we allow two agents to run different algorithms, the rendezvous problem can be solved using graph exploration. Specifically, one of the agents halts at the initial location and the other one traverses all the vertices. Hence the time complexity of 
graph exploration is a trivial upper bound for the rendezvous problem. 
In contrast, the half of the initial distance between two agents is a trivial 
lower bound for the problem. Since both of the bounds can be $\Theta(n)$ in a specific class of $n$-vertex 
instances (e.g., a ring network of $n$ vertices) 
the exploration-based approach is existentially optimal, but not universally
optimal. When the initial distance is small in terms of $n$, the approach 
based on graph exploration does not necessarily exhibit optimal algorithms.
However, due to the unavailability of the location information of other agents, 
achieving rendezvous without exploring all vertices is a highly non-trivial
challenge, even if we assume stronger capability of agents such as randomization,
asymmetry, and non-obliviousness.

\subsection{Contribution}
In this paper, we consider what instances and what computational power of models (oracles) admit efficient algorithms that do not use exhaustive search strategy, such as graph exploration. As we stated, the key characterization of the instances is distance of initial location of both agents. We consider the initial distance is small in terms of $n$, to avoid $\Omega(n)$ lower bound. In this setting, the meaning of "without exhaustive search" will be clear, namely presenting algorithms that achieves $o(n)$ rounds for rendezvous.  

In this paper, we consider an extreme variant of the rendezvous problem,
called the \emph{neighborhood rendezvous problem}, where two agents 
are initially located at \emph{two adjacent vertices} (i.e., initial distance 
one). This problem can be also seen as a generalized version of the rendezvous
problem in complete graphs~\cite{anderson1990rendezvous} because in that case 
any two agents always have distance one. Since the neighborhood rendezvous problem can be trivially solved 
in $O(\Delta)$ rounds ($\Delta$ is the maximum degree of the graph), the technical challenge
lies in the design of algorithms achieving rendezvous within $o(\Delta)$ rounds. As well as the algorithm shown in~\cite{anderson1990rendezvous}, we assume 
the rich capability of agents (i.e., randomized, asymmetric, and non-oblivious), 
unique vertex identifiers, and the availability of whiteboards placed at each 
vertex. 
 In addition, we assume that agents at a vertex $v$ can know the IDs of all $v$'s neighbors 
 (which is analogous to the \emph{KT1 model} in message passing systems~\cite{peleg2010distributed}).
 Specifically, we present two randomized algorithms. The first algorithm achieves rendezvous within $O \left( \frac{n}{\delta} \log^2 n + \sqrt{\frac{n \Delta}{\delta}}\log n \right)$ rounds with high probability for
graphs whose minimum degree is larger than $\sqrt{n}$.
Thus, this algorithm achieves fast rendezvous (i.e., sublinear of $\Delta$) in graphs with minimum degree 
$\delta = \omega(\sqrt{n} \log n)$. The second algorithm trades the use of whiteboards into the assumption of \emph{tight} naming of vertices, that is, the assumption that  
the largest vertex ID is $O(n)$. It achieves rendezvous within $O \left( \frac{n}{\sqrt{\delta}}\log^2 n \right)$ rounds with high probability\footnote{Throughout this paper, we say that an event $\mathcal{E}$ holds \emph{with high probability} if $\Pr[\mathcal{E}] \ge 1 - 1/n^{O(1)}$ holds}, 
and thus fast rendezvous is attained in the case of $\delta = \omega(n^{2/3} \log^{4/3} n)$. 

On the negative side, we also present the impossibility of sublinear-time rendezvous when we relax the assumptions. There lie four unconventional assumptions for our algorithm, which are bounded minimum degrees, the accessibility to neighborhood IDs, initial distance one, and randomization. Interestingly, the time lower bound of 
$\Omega(n)$ rounds for graphs of $\Delta = \Theta(n)$ is deduced even if we remove only one of them; this implies that our algorithm runs under a minimal assumption.  

\subsection{Related Work}
\label{subsec:related-work}
The solvability and complexity of the rendezvous problem is affected by many 
factors, such as synchrony, randomness of algorithms, graph classes, symmetry of agents, and so on. For that reason it is difficult to compare our results with past literature directly. Nevertheless, several results 
aim to achieve sublinear-time rendezvous explicitly or implicitly. 
Collins et al. \cite{collins2011synchronous} demonstrate that two agents with a common map (i.e., whole information of $G$), which
are initially placed with distance $d$, can achieve rendezvous deterministically 
within $O(d \log^2 n)$ rounds, they also show a nearly tight 
$\Omega(d\log n / \log \log n)$-round lower bound. Das et al. 
\cite{das2014rendezvous} assume that two agents can detect their distance, and present a deterministic rendezvous algorithm within $O(\Delta (d + \log l))$ rounds, where $l$ is the minimum value of the IDs of agents. 
It is also proven that any algorithm requires $\Omega(\Delta (d + \log l / \log \Delta))$ 
rounds in this model. 
The result by Anderson et al. \cite{anderson1990rendezvous} is the closest to our result in the sense
that it assumes no oracle such as maps and distance detection stated above. 
It considers the model of anonymous vertices with whiteboards, and presents a randomized algorithm that achieves rendezvous for complete graphs
in $O(\sqrt{n})$ expected rounds. As we mentioned, the neighborhood rendezvous 
problem can be seen as a relaxation of rendezvous in complete graphs, and thus
we can regard our result as the one extending the graph classes allowing fast rendezvous (using 
a stronger assumption of vertex identifiers).
There are also several studies \cite{miller2015tradeoffs,dereniowski2012drawing,miller2015fast} for 
achieving fast rendezvous using side information coming from oracles (so-called \emph{advice}). 
In this model, agents cannot see the whole map of $G$, but instead can know the (partial) information on their 
initial locations.

Due to the interest on hardness of symmetry breaking, the solvability of 
the rendezvous problem for ring networks has received much consideration in several different 
models~\cite{kranakis2003mobile,flocchini2004multiple,de2006asynchronous}. 
In this context, the analysis of complexity has not received much attention. The study of rendezvous in trees has focused on time and space complexities.
The paper by Baba et al.\cite{baba2013tcs} presents a linear-time (equivalently, 
$O(n)$ time) algorithm under the assumption that agents have $O(n)$ bits 
of memory, and the authors also show its optimality with respect to space in the class of linear-time algorithms. 
Czyzowicz et al. \cite{czyzowicz2014time} generalized this result, and presented an algorithm achieving 
rendezvous in $\Theta(n + n^2 / k)$ rounds for agents having $k$ bits of memory. Fraigniaud et al.~\cite{Fraigniaud2008disc}  
presents the rendezvous algorithm in trees with the optimal memory complexity ($\Theta(\log n)$ bits).
The feasibility of rendezvous in general graphs are also considered in 
several papers~\cite{Demarco2006tcs,Czyzowicz2010soda,czyzowicz2012meet,Bouchard2018ipl}. In paper \cite{czyzowicz2012meet}, the memory requirement for the rendezvous of uniform agents is considered, which presents that $\Theta(\log n)$ bits are necessary and sufficient for two agents in any anonymous graph. Recently, Miller et al.~\cite{miller2016time} consider the trade-offs between time and cost (the number of edges traversed by agents). 

The rendezvous problem allowing randomization is often considered as a part of 
the theory of random walks. The time taken for two tokens to meet at a common
vertex is called the \emph{meeting time}~\cite{tetali1991podc,Bshouty1999ipl}. The rendezvous problem 
in the analyses of Markov chain theory is also considered in the context of operations research~\cite{anderson1990rendezvous,
alpern2002rendezvouslabeled,abbas2008probabilistic, yu1996agent,dani2016codes,weber2012optimal}.

A comprehensive overview of the rendezvous problem can be found in the books by Alpern and Gal~\cite{alpern2006theory} and Alpern et al.~\cite{alpern2013search}, and several 
surveys~\cite{kranakis2006mobile,pelc2012deterministic,alpern2002rendezvous}.

    \section{Preliminaries}
    \label{sec:preliminaries}
    \subsection{Model and Notations}
    In this paper, we consider the rendezvous problem of two agents in any 
    undirected graph $G = (V, E)$ of $n$ vertices. Each vertex in $G$ has a distinct integer identifier in $[0, n'-1]$, where $n'$ satisfies $n' \geq n$ and $n' = n^{O(1)}$. The value of $n'$ is available to each agent. 
    We denote the identifiers of $n$ vertices by $v_0, v_1, \dots, v_{n-1}$. The 
    minimum and maximum degrees of $G$ are respectively denoted by 
    $\delta_G$ and $\Delta_G$. For any vertex $v$, $N_G(v)$ represents the set of vertices adjacent to $v$, i.e., $N_G(v) = \{ v' ~|~ (v, v') \in E \}$. We define
    $N^+_G(v) = N_G(v) \cup \{ v \}$, and also define $N_G(X) = \bigcup_{v \in X} N_G(v)$ and $N^+_G(X) = N_G(X) \cup X$ 
    for any vertex set $X \subseteq V$. We often omit subscript $G$ if it is clear from the context. 
    
    In the system, two computing entities, called \emph{agents}, are placed at two vertices in $G$, which are modeled as probabilistic random access machines. 
    The two agents have distinct names denoted by $a$ and $b$ respectively, and can
    exhibit asymmetric behavior in executions, that is, they can run two 
    different algorithms. Agents are equipped with memory space as their internal states. 
    While we do not assess any assumption on time/space complexity for internal computation of agents, 
    our proposed algorithms terminate within polynomial time, and use $O(n \log n)$-bit memory.
    We denote by $M \subseteq \{ 0, 1 \}^*$ the set of possible internal states of two agents.
    When two agents visit the same vertex, they are aware of the presence of each other. 
    On neighborhood knowledge, we define
    the \emph{local port numbering} of each vertex $v_i$, which is a bijective function 
    $\hat{P}_{v_i} : [0, |N(v_i)| - 1] \to N(v_i)$. We also define the \emph{accessible local port 
    number} $P_{v_i} : [0, |N(v_i)| - 1] \to \mathbb{N}$. Agents can see only $P_{v_i}$ and have no access to $\hat{P}_{v_i}$. The model supporting the access to neighborhood IDs is defined as the assumption that 
    $\hat{P}_{v_i}$ and $P_{v_i}$ are the same function for any $v_i \in V$. On the lower-bound side, we also 
    consider the case where each agent has no access to its neighborhood IDs. It is defined as the model 
    such that $P_{v_i}$ for any $v_i$ is the identity mapping from 
    $[0, |N(v_i)| - 1]$ to $[0, |N(v_i)| - 1]$ (i.e., it does not provide any information of $\hat{P}_{v_i}$).

    Each vertex is equipped with a memory space called \emph{whiteboards}, and 
    an agent at vertex $v$ can access/write to the whiteboard of $v$
    in its internal computation. Formally, we define $W \subseteq \{0, 1\}^{\ast}$ to be the set of 
    possible contents written in each whiteboard. A state of all the whiteboards 
    in $G$ is represented by an $n$-dimensional vector $W^n$ indexed by elements in $V$. While we have no assumption on the size of each whiteboard, 
    $O(\log n)$ bits per vertex suffice for our algorithms. 
    
    Executions of two agents follow synchronous and discrete time steps $t = 0, 1, 2, \dots$ called \emph{rounds}. In every round, an agent at vertex $v$ 
    either stays at the present location or moves to one of its neighbors. An 
    algorithm $\mathcal{A}$ determines which action to take based on the information stored
    in its internal memory, IDs in $N^+(v)$ through the access to $P_{v}$, and the contents of the whiteboard at $v$. We assume that a movement to a neighbor necessarily completes within the current round. In other words, we do not consider the situation where agents are located on edges at the beginning of each round. At each round, agents can modify the whiteboards of their current vertices\footnote{Strictly, we need to define formally the behavior of agents when they are located at the same vertex and attempt to modify the (common) whiteboard. In the rendezvous problem of two agents, however, such a 
    case can be seen as the completion of the algorithm without loss of generality. Thus, we do not care about simultaneous and parallel write operation for the same whiteboard.}.
    Formally, an algorithm is a function $\mathcal{A}: 
    \{ a, b \} \times M \times V \times 2^{\mathbb{N}} \times W \times \{0, 1\}^* \rightarrow M \times \mathbb{N} \times W$. The inputs respectively correspond to the ID of 
    the agent, its internal memory, the IDs of its current location and neighbors (with respect to accessible 
    port numbering functions), 
    the content of the whiteboard at the current location, and 
    random bits. The outputs correspond to the internal state of the agent after the computation, the destination 
    in the following movement (with respect to accessible local port numbers), and the content of the whiteboard left at the current vertex. Note that deterministic algorithms (only used in Section \ref{subsec:deterministiclowerbound}) are defined as the ones such that its behavior is independent of random bits. A \emph{configuration} $C$ at round $t$ is a tuple in $C \in (V \times M)^2 \times W^n $. 
    An \emph{execution} is an infinite sequence of configurations 
    $C_0, C_1, C_2, \dots$. Precisely, letting $v_i^z$ be the location of agent 
    $z \in \{a, b\}$ at round $i$, $m_i^z$ be the internal memory of agent $z$ at round $i$, 
    and $w_{i}^j$ be the the content of the whiteboard of vertex $v_j$ at round $i$,
    a configuration $C_i$ is described as $C_i = (v_i^a, m_i^a, v_i^b, m_i^b, w_i^0, \dots, w_i^{n-1} )$. For any $i \in \mathbb{N}$, every execution must satisfy 
    the following conditions: For any $j \in V \setminus \{ v^a_i, v^b_i \}$ 
    $w_j = w'_j$ holds. For each $i$, there exists $B^a_i, B^b_i
    \in \{ 0,1 \}^*$ such that $\mathcal{A}(a, m_i^a, v_i^a, P_{v_i^a}, w_{v_i^a}, B^a_i) = (m_{i+1}^a, \hat{P}^{-1}_{v_i^a}(v_{i+1}^a), w_{v_i^a}^i) $ and $\mathcal{A}(b, m_i^b, v_i^b, P_{v_i^a}, w_i^{v_i^b}, B^b_i) = (m_{i+1}^b, \hat{P}^{-1}_{v_i^b}(v_{i+1}^b), w_i^{v_i^b}) $ hold, where $P^{-1}_{v_i^a}$ and 
    $P^{-1}_{v_i^b}$ are the inverse mappings of $P_{v_i^a}$ and 
    $P_{v_i^b}$ respectively.

    \subsection{Rendezvous Problem}
    In the rendezvous problem, two agents initially located at two different vertices are required to visit the same 
    vertex simultaneously and halt. Formally, an algorithm completes rendezvous at round $t$ if the two agents are located at
     the same vertex at the beginning of that round\footnote{In the synchronous system, we can assume that once two agents meet at 
    a vertex then they halt without loss of generality. That is, agents that complete rendezvous at round $t$ also
    complete rendezvous at any round $t' > t$.}. This paper considers the rendezvous
    problem with the constraint on initial locations of agents and graph parameters.
    \begin{definition}[Specific Rendezvous]
        For graph $G = (V, E)$, let $I \subseteq V \times V$ be a possible set of initial locations $(v_0^a, v_0^b)$ of two agents. We say that an algorithm $\mathcal{A}$ solves the rendezvous problem for an instance $(G, I)$ with probability $p$ 
        within $t$ rounds, if for any $(v_0^a, v_0^b) \in I$, the execution of 
        $\mathcal{A}$ in $G$ completes rendezvous at round $t$ with probability $p$. Moreover, letting $\mathcal{I} = \{ (G_0, I_0), (G_1, I_1), \dots \}$ be 
        a (possibly infinite) class of instances, 
        we say that an algorithm $\mathcal{A}$ solves the rendezvous problem for class $\mathcal{I}$ with probability $p$ within $f(n)$ rounds for some non-decreasing function
        $f: \mathbb{N} \rightarrow \mathbb{N}$ if for every instance $((V, E), I) \in \mathcal{I}$, algorithm $\mathcal{A}$ solves the rendezvous problem with probability $p$ within $f(|V|)$ rounds. 
    \end{definition}

    In this paper we are interested in the case where the distance between two 
    initial locations of agents is upper bounded by $d$. For any graph $G$ we define $I^G_{d} = \{ (v, v') ~|~ dist_G(v, v') \le 
    d \}$. In addition, we also define the class $\mathcal{G}(\hat{\Delta}(n), \hat{\delta}(n))$ for functions $\hat{\delta}:\mathbb{N} \rightarrow \mathbb{N}$, $\hat{\Delta}: \mathbb{N} \rightarrow \mathbb{N}$ as the set of graphs $G = (V, E)$ such that $\delta_G \geq \hat{\delta}(|V|)$ and $\Delta_G \leq \hat{\Delta}(|V|)$ hold.
    The \emph{$(\hat{\Delta}(n), \hat{\delta}(n), d)$-rendezvous problem} is defined as that for the instance 
    class $\mathcal{I}_{d} = \{ (G, I^G_{d}) \mid G \in \mathcal{G}(\hat{\Delta}(n), \hat{\delta}(n))  \}$. In particular, we focus on the instance class $\mathcal{I}_{1}$ in Section \ref{sec:algorithm} and \ref{sec:discussion}. In Section \ref{sec:impossibility} we show the lower bounds on the problem for $\mathcal{I}_2$.

    \section{Rendezvous Algorithm} 
    \label{sec:algorithm}
    \subsection{Algorithm Overview}

    In this section, we present an overview of our rendezvous algorithm. For ease of presentation, we assume that each agent has the precise values of $\delta$ and $\log n$, but it is not essential. Those values can be replaced with their constant-factor approximate values without increasing the asymptotic running time.
    A constant factor approximation
    of $\log n$ can be estimated from the upper bound $n'$ of vertex IDs. The approximation of $\delta$ can be obtained by standard doubling estimation, explained in Section \ref{sec:discussion}.

    
    First, we introduce several definitions and terminologies used in the following argument.
    
    \begin{definition}[$\alpha$-heaviness, $\alpha$-lightness]
        For any $T \subseteq V$, $v \in V$, and $\alpha \in \mathbb{R}^+$, 
        $v$ is called \emph{$\alpha$-heavy} for $T$ if $|T \cap N^+(v)| \ge \alpha$ holds~\footnote{$\mathbb{R}^+$ is the set of all positive real values.}.
        Similarly we say that $v$ is \emph{$\alpha$-light} for $T$ if 
        $|T \cap N^+(v)| < \alpha$ holds.
    \end{definition}
    
    The following proposition is a trivial fact deduced from the definition above.
    
    \begin{proposition}
        \label{prop:heaviness}
        Let $v \in V$ be an $\alpha$-heavy vertex for $T \subseteq V$. 
        For any $T'$ such that $T' \supseteq T$ holds, $v$ is also $\alpha$-heavy for $T'$.
    \end{proposition}
    
    Given a vertex set $T \subseteq V$ and $\alpha \in \mathbb{R}^+$, we define 
    $H_{\alpha} (T), L_{\alpha}(T) \subseteq V$ as the sets of vertices that are respectively $\alpha$-heavy and $\alpha$-light for $T$.
    
    \begin{definition}[$(z, \alpha, \beta)$-dense condition]
        Given $z \in \{ a, b \}$, $T \subseteq V$, and $\alpha, \beta \in \mathbb{R}^+$, $T$ is called \emph{$(z, \alpha, \beta)$-dense} if the
        following three conditions hold:
        \begin{itemize}
            \item $v_0^z \in T$,
            \item for any $w \in T$, $dist_G (v_0^z, w) \le \beta$, and
            \item $N^+(v_0^z) \subseteq H_{\alpha}(T)$.
        \end{itemize}
    \end{definition}
    
    The main idea of our rendezvous algorithm is that agent $a$ constructs an 
    $(a, \delta/8, 2)$-dense vertex set $T^a$. Since $v^b_0 \in N^+(v_0^a)
    \subseteq H_{\delta/8}(T^a)$, $v_0^b$ is an $(\delta / 8)$-heavy vertex for $T^a$.
    Then a sublinear number of random vertex samplings from $T^a$ by agent $a$ and 
    those from $N(v_0^b)$ by $b$ ensure that a vertex is commonly sampled with high probability. 
    In this sampling process, agent $b$ leaves the ID of $v^b_0$ 
    at the whiteboards of all the sampled vertices. When agent $a$ visits the common sample,
    it knows the initial location of $v^{b}_0$. Then agent $a$ 
    moves to $v^b_{0}$ and meets $b$.
    
    In the following argument, we divide our algorithm into two sub-algorithms. The
    first one, called \textsf{Main-Rendezvous}, achieves rendezvous provided that agent $a$ knows an $(a, \delta/8, 2)$-dense set $T^a \subseteq N^+(N^+(v^a_0))$.
    The second sub-algorithm 
    is for agent $a$ to construct such an $(a, \delta/8, 2)$-dense set 
    $T^a$, which is called \textsf{Construct}. The combination
    of these two sub-algorithms yields the algorithm we claim. 
    
    \subsection{Rendezvous with $T^a$}
    \label{sec:dense-rendezvous}
    We present the algorithm \textsf{Main-Rendezvous}, which solves 
    the rendezvous problem using the initial knowledge of 
    an $(a, \delta/ 8, 2)$-dense set $T^a \subseteq N^+(N^+(v^a_0))$ by agent $a$. Here
    the ``knowledge'' implies that (1) $a$ has the list of all vertices in $T^a$ 
    in its memory, 
    and (2) also has the shortest paths to all vertices in $T^a$ from $a$'s initial location\footnote{Since the length of these shortest paths are at most two by the definition of 
    $(a, \delta/8, 2)$-dense sets, the space for storing this information 
    is asymptotically same as the space for the list of vertices. }.
    The pseudocode of \textsf{Main-Rendezvous} is presented in Algorithm~\ref{alg:rendezvous}.
    First, agent $a$ samples a vertex 
    $v$ in $T^a$ uniformly at random, and visits there. At vertex $v$, $a$ 
    checks if $b$ has written the ID $v^b_0$ in the whiteboard of $v$. If so, then
    $a$ moves to $v^b_0$ and halts. The agent $b$ iteratively visits a vertex $u$ 
    in $N^+(v^b_0)$ chosen uniformly at random, and writes down the ID of $v^b_0$ 
    into the whiteboard of $u$. If it meets $a$ at vertex $v^b_0$, then the algorithm terminates.
    We present the following lemma for the correctness of \textsf{Main-Rendezvous}.

    \begin{algorithm}[t]
        \caption{\textsf{Main-Rendezvous} : Rendezvous with $T^a$}
        $w(v)$ : whiteboard at vertex $v$. Initially $w(v) = \perp$ for all $v \in V$ \\
        $q_a, q_b$ : local variables of agents $a$ and $b$ \\
        \textbf{Operations of Agent $a$}
        \begin{algorithmic}[1]
            \label{alg:rendezvous}
            \setlength{\algorithmicindent}{15pt}
            \STATE construct $T^a$ satisfying $(a, \delta/8, 2)$-dense condition
            \REPEAT
                \STATE choose $v$ in $T^a$ uniformly at random, and move to $v$
                \STATE $q_a \leftarrow w(v)$
                \STATE return to $v^a_0$
            \UNTIL{$q_a \neq \perp$}
            \STATE visit $q_a$ and halt
        \end{algorithmic}
        \textbf{Operations of agent $b$}
        \begin{algorithmic}[1]
            \setlength{\algorithmicindent}{15pt}
            \REPEAT
                \STATE move to $v \in N^+(v_0^b)$ chosen uniformly at random
                \STATE $w(v) \leftarrow v_0^b$
                \STATE return to $v^b_0$
            \UNTIL{achieve rendezvous}
        \end{algorithmic}
    \end{algorithm}
    
    \begin{lemma}
        \label{lemma:dense-rendezvous}
        Let $G = (V, E)$ be any graph such that $\delta_G \ge \sqrt{n}$ holds. Suppose that agent $a$ constructs an 
        $(a, \delta / 8, 2)$-dense set $T^a$ in $t_a$ rounds. Then, Algorithm \textsf{Main-Rendezvous} completes rendezvous within $t_a + O\left(\sqrt{\frac{n \Delta}{\delta}} \log n\right)$ rounds with high probability.
    \end{lemma}
    \begin{proof}
        We say that a vertex $v \in N^+(v^b_0) \cap T^a$ is \emph{informed} at round $t$ if 
        $w(v) = v^b_0$ at $t$, and define $Z_t \subseteq N^+(v^b_0) \cap T^a$ as 
        the set of all informed vertices at $t$. 
        Let $h = \lfloor (1/16)\sqrt{n \delta / \Delta} \rfloor $ for short. We first show that $|Z_t| \geq h$ 
        holds for $t \geq t_a +  8\sqrt{n\Delta/\delta}\log n$.  
        Let $t_i$ be the first time that $Z_{t_i} \geq i$ holds, and $X_i$ be $X_i = t_i - t_{i - 1}$ 
        ($1 \leq i \leq h$). 
        By the assumption of $\delta > \sqrt{n}$, we have the following inequality.
        \[ 
            h = \left\lfloor \frac{1}{16} \sqrt{\frac{n \delta}{\Delta}} \right\rfloor 
            \le \frac{1}{16} \sqrt{ n } \le \frac{1}{16} \delta < |N^+(v_0^b) \cap T^a|.
        \] 
        For any $1 \leq i \leq h $, the variable $X_i$ follows the geometric distribution with success 
        probability $p_i = (|N^+(v_0^b) \cap T^a| -i +1)/|N^+(v_0^b)|$.
        Then we have
        \begin{align*}
        \expect{X_i} &= \frac{|N^+(v_0^b) |}{|N^+(v_0^b) \cap T_a| -i + 1} \\
        &\leq \frac{|N^+(v_0^b)|}{|N^+(v_0^b) \cap T^a| - h + 1}. 
        \end{align*}
        This deduces the following bound.
        \begin{align*}  
            \mathrm{E} \left[ t_h \right] = t_a + \mathrm{E} \left[ \sum_{i=2}^{\lfloor h \rfloor} X_i \right] &\le t_a + \sum_{i=2}^{h} \frac{|N^+(v_0^b)|}{|N^+(v_0^b) \cap T^a| - h}  \\
            &\le t_a + h \frac{(\Delta + 1)}{\delta / 16} \\ 
            &\le t_a + \left\lfloor \frac{1}{16} \sqrt{\frac{n \delta}{\Delta}} \right\rfloor  \frac{16 (\Delta + 1)}{\delta}  \\
            &\le t_a + 2 \sqrt{\frac{n \Delta}{\delta}}.
        \end{align*}
        By Markov's inequality, the probability of $|Z_t| < h$ for $t = t_a + 4 \sqrt{n \Delta / \delta}$ is 
        at most $1/2$. Thus the probability of $|Z_t| < h$ for $t = t_a + 8 \sqrt{n \Delta / \delta}\log n$
        is at most $1/n^2$.
        
        Assume that $|Z_t| \ge h$ holds for $t = t_a +  8\sqrt{n\Delta/\delta}\log n$. At $t$ or later, the probability that 
        agent $a$ visits an informed vertex is at least $h / |T^a|$. Bounding the tail bound using Markov's inequality, we can conclude that agent $a$ visits at least one informed vertex by the time $t_a + O\left( \sqrt{\frac{n \Delta}{\delta}} \log n \right)$ with probability $1 - 1/n^2$ or more. That is, two agents
        meet within  $t_a + O\left( \sqrt{\frac{n \Delta}{\delta}} \log n \right)$ rounds with probability at least 
        $1 - O(1/n^2)$. Hence, the lemma is proven. 
    \end{proof}
    
    \subsection{Construction of $T^a$}
    \label{sec:constructT}
    In what follows, we simply say that a vertex is heavy or light if it is 
    $\delta/8$-heavy or $\delta / 2$-light respectively.
    By Lemma~\ref{lemma:dense-rendezvous}, it  suffices that agent $a$ constructs
    a $(a, \delta / 8, 2)$-dense set $T^a$ to achieve rendezvous. The algorithm 
    \textsf{Construct} takes the role of constructing $T^a$, which utilizes a subroutine 
    called \textsf{Sample}. The pseudocode of \textsf{Sample} and \textsf{Construct}  
    are presented in Algorithm \ref{alg:sample} and \ref{alg:construct} respectively.
    In algorithm \textsf{Construct}, agent $a$ manages a set $S^a \subseteq 
    N^+ (v^a_0)$, and iteratively adds a vertex to 
    $S^a$. In the following argument, we refer to the process of adding the 
    $i$-th vertex to $S^a$ as the \emph{$i$-th iteration}. Eventually, the 
    algorithm outputs $N^+(S^a)$ as the constructed set $T^a$ when it satisfies the termination condition (which is explained later). 
    Let $S^a_i$ be the set stored in $S^a$ at the beginning of the $i$-th iteration, 
    and $x_i$ be the vertex added in the 
    $i$-th iteration. The principle of choosing $x_i$ is very simple: Agent $a$ selects a vertex $x_{i}$ such that the volume of 
    $N^+(x_i) \setminus N^+(S_i^a)$ is large. Specifically, it searches a vertex 
    $w \in N^+(v_0^a)$ that is light for $N^+(S_i^a)$. If such 
    a vertex exists, it is added to $S_i^a$ as $x_{i}$. Otherwise, any vertex in 
    $N^+(v^a_0)$ is heavy for $N^+(S^a_i)$, i.e., $N^+(v^a_0) \subseteq 
    H_{\delta /8 }(N^+(S^a_i))$.
    This implies that $N^+ ( S^a_i )$ satisfies 
    $(a, \delta/8, 2)$-dense condition, and the algorithm can return it as $T^a$.
    Adding a light vertex to $S^a_i$ increases the cardinality of $N^+(S^a)$ by at least $\Theta(\delta)$, 
    and thus the algorithm 
    \textsf{Construct} obviously terminates within $O(n/\delta)$ iterations (because 
    if $N^+(S^a) = V$ holds, any vertex becomes heavy for $N^+(S^a)$).
    
    For expanding $S^a_i$ by adding a light vertex, the 
    algorithm has to check the heaviness of each vertex in $N^+(v_0^a)$ 
    (for $N^+(S_i^a)$). The algorithm \textsf{Sample} takes this role. More 
    precisely, the run of 
    \textsf{Sample}$(\Gamma, \alpha)$ probabilistically checks 
    whether or not each vertex in $N^+(v_0^a)$ is $\alpha$-heavy for $\Gamma$  within $O(|\Gamma|/\alpha)$
    rounds. The algorithm outputs the vertex set consisting of the vertices concluded as $\alpha$-heavy for $\Gamma$. A straightforward approach of identifying $x_{i}$ in the 
    construction of $T^a$ is to run \textsf{Sample}$(N^+(S_i^a), \delta / 8)$ in 
    every iteration. However, then the total running time of \textsf{Construct}
    becomes $O( (n / \delta)^2)$ rounds. To save time, our algorithm 
    finds a light vertex $x_i$ using the following two-step strategy:
    
    \begin{itemize}
        \item (Step 1) \textbf{Optimistic decision}: In the $i$-th iteration,
        agent $a$ runs \textsf{Sample}$(\Gamma, \delta / 8)$
        for $\Gamma = N^+(S_i^a) \setminus N^+(S_{i-1}^a)$.
        If it detects that a vertex $u \in N^+(v_0^a)$ is heavy for $\Gamma$, 
        Proposition~\ref{prop:heaviness} guarantees that $u$ is heavy for 
        $N^+(S_i^a) \supseteq \Gamma$. On the other hand, vertex $u$ can be heavy for $\Gamma$ even if 
        the algorithm
        says that $u$ is light. Then adding a vertex $u$ as $x_{i}$ prevents the algorithm 
        from working correctly as intended. 
        \item (Step 2) \textbf{Strict decision}: To resolve the matter of step 1, 
        agent $a$ checks if the candidates of $x_{i}$ are actually light for 
        $N^+(S_i^a)$. More precisely, the agent samples $\Theta(\log n)$ vertices 
        uniformly at random from the set output by the run of 
        \textsf{Sample}$(\Gamma, \delta / 8)$, and then it checks the heaviness of each sample $v$ by actually visiting there and computing $|N^+(S_i^a) \cap N^+(v)|$. If the agent finds a light vertex from the $\Theta(\log n)$ samples, that vertex is selected as $x_{i}$. Otherwise, it finds that a constant fraction of whole candidates for $x_{i}$ in the optimistic decision is heavy for $N^+(S_i^a)$ with high probability. Then the agent runs \textsf{Sample}$(N^+(S_i^a), \delta / 8)$ for strict checking.
        If a vertex $u$ is found light for $N^+(S_i^a)$, the agent selects $u$ as $x_{i}$. Otherwise, the algorithm 
        terminates.
    \end{itemize}
    
    In the following argument, we refer to the runs of \textsf{Sample} in step 1 and 2 as \emph{optimistic/strict runs} of \textsf{Sample}$(\Gamma, \alpha)$ respectively. Since the running time of each optimistic run depends on the size of 
    difference set $N^+(S_i^a) \setminus N^+(S_{i-1}^a)$, the total sum of the 
    running time incurred by optimistic runs is $O((n \log n ) / \delta)$. While 
    each strict run of \textsf{Sample} needs at most $O((n \log n) / \delta)$ rounds,
    we can show that strict runs are executed at most $O(\log n)$ times. 
    It comes from the two facts that 1) one strict run corrects the identification of a constant fraction of heavy vertices in $N^+(S_i^a)$ which are wrongly identified as light ones, and 2) a vertex identified as a heavy one is never identified as light. Consequently the total running time of \textsf{Construct} is bounded by
    $O(n \log^2 n / \delta)$ steps.
    We explain the details of \textsf{Sample}$(\Gamma, \alpha)$ and \textsf{Construct} 
    in the following paragraphs.
    
    \subsubsection{\textsf{Sample}$(\Gamma, \alpha)$}
    For the decision of lightness/heaviness of each vertex in 
    $N^+(v_0^a)$ for $\Gamma$, this algorithm conducts random samplings and visits. The agent uses an 
    array $C \subseteq \mathbb{Z}^{|N^+(v_0^a)|}$, which counts for each $u \in N^+(v_0^a)$ the number of visited vertices having $u$ as a neighbor. The initial value of $C[u]$ for each $u \in N^+(v_0^a)$ is $C[u] = 0$. Let $l$ be a threshold value $l = \lceil 150 \ln n \rceil$. 
    In the run of \textsf{Sample}$(\Gamma, \alpha)$, the agent repeatedly visits a vertex 
    $v$ in $\Gamma$ chosen uniformly at random (with duplication) $96 \lceil |\Gamma| ( \ln n )/ \alpha \rceil$ times. At the visited vertex $v$, it increments $C[u]$ for 
    each vertex $u$ in $N^+(v_0^a) \cap N^+(v)$ (for this process, the agent carries the information of $N^+(v_0^a)$). After processing all samples, 
    the agent concludes that $u$ is heavy for $\Gamma$ if $C[u] \ge l$ holds, or 
    light otherwise. The algorithm outputs the vertex set $H'$ consisting of the vertices concluded as a heavy one. 
    
    \subsubsection{\textsf{Construct}}
    In this algorithm agent $a$ has the following sets as its internal variables:
    $S_i^a$, $R_i$, $H_i$, and $\mathit{NS}_i^a$. The subscript $i$ corresponds to the number of 
    iterations in the algorithm. The set $R_i$ is a set of candidates for $x_{i}$. The set $H_i$ stores the vertices that turned out to be $(\delta / 8)$-heavy for $N^+(S_i^a)$ at the $i$-th iteration. The variable $\mathit{NS}_i^a$ keeps track of the set $N^+(S^a_i)$. The initial value of these sets are $S_1^a = \{ v_0^a \}$, $R_1 = N^+(v_0^a)$, $H_1 = \emptyset$, and $\mathit{NS}_i^a = N^+(v_0^a)$ respectively. The agent $a$ iterates the following operations until $R_i = \emptyset$. First, the agent executes the optimistic run of \textsf{Sample}$(N^+(S_i^a) \setminus N^+(S_{i-1}^a), \delta / 8)$, and for the returned set $H'$ it updates $H_{i}$ and $R_{i}$ with $H_{i+1} \leftarrow H_i \cup H'$ and $R_i \leftarrow N^+(v_0^a) \setminus H_{i+1}$. Based on the updated set $R_{i+1}$, the agent randomly chooses $\lceil 4 \log n \rceil$ vertices from 
    $R_{i+1}$ and visits each sampled vertex. If a visited vertex is actually light for $N^+(S_i^a)$ (this is checked by using the information of $\mathit{NS}_i^a$), then the agent adds it to $S_i^a$ as $x_i$. Otherwise, (i.e., all of the vertices are heavy for $N^+(v_0^a)$), then the agent executes the strict run of \textsf{Sample}$(N^+(S_i^a), \delta / 8)$ and updates the set $H_{i+1}$ and $R_{i+1}$ in the same way as the optimistic run. After that, the agent selects any vertex in $R_{i+1}$ and adds it to $S_i^a$.   
    
    \begin{algorithm}[t]
        \caption{\textsf{Sample}$(\Gamma,  \alpha)$}
        $l$: threshold value $l = \lceil 150 \ln n \rceil$
        \begin{algorithmic}[1]
            \label{alg:sample}
            \setlength{\algorithmicindent}{15pt}
            \FOR{$i=1$ to $96 \left\lceil \frac{ |\Gamma|  \ln n }{\alpha} \right\rceil$}
                \STATE choose a vertex $v$ in $\Gamma$ uniformly at random \STATE visit $v$
                \FORALL{$u \in N^+(v) \cap N^+(v_0^a)$}
                    \STATE $C[u]++$
                \ENDFOR
            \ENDFOR
            \FORALL{$u \in N^+(v_0^z)$}
                \IF{$C[u] \ge l$}
                    \STATE $H' \leftarrow H' \cup \{ u \}$
                \ENDIF
            \ENDFOR
            \STATE return $H'$
        \end{algorithmic}
    \end{algorithm}
    
    \begin{algorithm}[t]
        \caption{\textsf{Construct}} 
    \begin{algorithmic}[1]
        \label{alg:construct}
        \setlength{\algorithmicindent}{15pt}
        \WHILE{$R_i \neq \emptyset$}
        \STATE $H' \leftarrow \textsf{Sample}(N^+(S_i^a) \setminus N^+(S_{i-1}^a), \delta/8)$;
        \STATE $H_{i+1} \leftarrow H_i \cup H'$;\STATE $R_{i+1} \leftarrow N^+(v_0^a) \setminus H_{i+1}$;
        \IF{$R_{i+1} \neq \emptyset$}
        \FOR{$j = 1$ to $\lceil 4 \log n \rceil$}
            \STATE choose $u_j \in R_{i+1}$ uniformly at random;
            \STATE visit $u_j$; 
            \STATE compute $|N^+(S^a_i) \cap N^+(u_j)|$ using $\mathit{NS}_i^a$;
            \IF{$u_j$ is $\delta / 2$-light for $N^+(S^a_i)$}
                \STATE $x_i \leftarrow u_j$
                \STATE $S_{i+1}^a \leftarrow S_i^a \cup \{ x_{i} \}$; 
                \STATE $R_{i+1} \leftarrow R_{i+1} \setminus \{ x_i \}$; 
                \STATE break;
            \ENDIF
        \ENDFOR
        \IF{Each $u_j$ is $\delta / 2$-heavy for $N^+(S_i^a)$}
            \STATE $H'  \leftarrow \textsf{Sample}(N^+(S^a_i), \delta / 8)$;
            \STATE $H_{i+1} \leftarrow H_{i+1} \cup H'$; 
            \STATE $R_{i+1} \leftarrow N(v_0^a) \setminus H_{i+1}$;
            \IF{$R_{i+1} \neq \emptyset$}
            \STATE choose any vertex $x_{i} \in R_{i+1}$;
            \STATE $S_{i+1}^a \leftarrow S_{i}^a \cup \{ x_{i} \}$; 
            \STATE $\mathit{NS}_i^a \leftarrow \mathit{NS}_i^a \cup N^+(x_i)$; 
            \STATE $R_{i+1} \leftarrow R_{i+1} \setminus \{ x_{i} \}$;  
            \ENDIF
        \ENDIF
        \ENDIF
        \STATE $i \leftarrow i+1$
        \ENDWHILE
        \STATE return $N^+(S_i^a)$
    \end{algorithmic}    
    \end{algorithm}
    
    \subsection{Correctness Proof of Algorithm \textsf{Sample}$(\Gamma, \alpha)$}
    Lemma \ref{lma:detect-heavy} below shows that the algorithm 
    \textsf{Sample}$(\Gamma, \alpha)$ probabilistically checks if a vertex $u \in N^+(v_0^a)$ is approximately heavy or light 
    for $\Gamma$. 
    
    \begin{lemma}
        \label{lma:detect-heavy}
        Let $\alpha > 0$ and $\Gamma \subseteq N^+(v_0^a)$ satisfy $|
        \Gamma| \ge \alpha$. The following statements 
        hold for any $u \in N^+(v_0^a)$ and the 
        output set $H'$ of \textsf{Sample}$(\Gamma, \alpha)$ 
        with probability at least $1 - 1/n^8$:
        \begin{enumerate} 
        \item If $u \in H'$ then $u$ is $\alpha$-heavy 
        for $\Gamma$. 
        \item if $u \in N^+(v_0^a) 
        \setminus H'$ then $u$ is $4 \alpha$-light for 
        $\Gamma$.
        \end{enumerate}
    \end{lemma}
    
    \begin{proof}
        We prove that 1) if $u \in N^+(v_0^a)$ is $\alpha$-light for $\Gamma$, then 
        after the execution of the algorithm,  $C[u] < 
        l$ holds with high probability., and 2) if the vertex $u$ is $4 
        \alpha$-heavy then $C[u] \ge l$ with high probability. This trivially implies the lemma.
        Consider the proof of the first statement.  
        Suppose that $u$ is $\alpha$-light for $\Gamma$. 
        Then we have $|N^+(u) \cap \Gamma| < \alpha$. 
        Let $X_1$ be the random variable corresponding to the value stored in $C[u]$ 
        after the execution of \textsf{Sample}$(\Gamma, \alpha)$. Since $X_1$ follows the binomial distribution $B(m, p)$ with parameter $p = |N^+(u) \cap \Gamma| / |\Gamma| < \alpha /|\Gamma|$ and $m = 96 \lceil ( |\Gamma| \ln n ) / \alpha \rceil$, $\expect{X_1} \le 96 \lceil ( |\Gamma| \ln n ) / \alpha \rceil \cdot \alpha / |\Gamma| \le (96 \ln n) + 1$ holds. Let $\mu_1 = (96 \ln n) + 1$ for short. 
        Using Chernoff bound, we have 
        \[
            \Pr[X_1 \ge l] \le \Pr[X_1 \ge (1 + 1/2) \mu_1] \le \mathrm{e}^{- \mu_1 / (3 \cdot 2^2)} \le 1/n^{8}.    
        \]
    
        We next consider the second statement. Suppose that $u$ is $4 \alpha$-heavy for $\Gamma$. Then we have $|N^+(u) \cap \Gamma| \ge 4 \alpha$. Similarly, with the first proof, we define the random variable $X_2$
        corresponding the value of $C[u]$ after the execution of the algorithm.
        Since it follows the binomial distribution $B(m, p)$ with the same parameter as the first proof, we have $\expect{X} \ge 96 \lceil ( |\Gamma| \ln n ) / \alpha \rceil \cdot ( 4 \alpha / |\Gamma|) \ge 96 ( ( |\Gamma| \ln n ) / \alpha  ) \cdot ( 4 \alpha / |\Gamma| ) \ge 384 \ln n$. Letting $\mu_2 = 384 \ln n$, Chernoff bound provides the following inequality.
        \[
            \Pr[X_2 \le l] \le \Pr [(1 + 1/2)\mu_2] \le \mathrm{e}^{-\mu_2 / (3 \cdot 2^2)} \le 1/n^{8}.  
        \]
        Thus, the lemma is proven.
    \end{proof}
    
    The next corollary immediately implies the correctness of algorithm 
    \textsf{Sample}$(\Gamma, \alpha)$, which is obtained by Lemma~\ref{crl:detect-heavy} and the standard union-bound argument.
    
    \begin{corollary}
        \label{crl:detect-heavy}
        Consider any call of \textsf{Sample}$(\Gamma, \alpha)$. If $|\Gamma| \ge \alpha$, then $H' \subseteq H_{\alpha}(\Gamma)$ and $N^+(v_0^a) \setminus H' \subseteq L_{4 \alpha}(\Gamma)$ hold with probability at least $1 - 1/n^7$. 
    \end{corollary}
    Note that the running time of the algorithm \textsf{Sample}$(\Gamma, \alpha)$ is $O(\frac{\Gamma \ln n}{\alpha})$.
    
    \subsection{Correctness Proof of Algorithm \textsf{Construct}}

    Now we turn to the analysis of the algorithm \textsf{Construct}. Our first goal of this analysis is to show that the algorithm \textsf{Construct} constructs a desired $(a, \delta/8, 2)$-dense set $T^a$ in $O(n/\delta)$ iterations.
    As we stated at the description of the algorithm (in section \ref{sec:constructT}), the key observation for this goal is that in each iteration the algorithm adds a light vertex $x_i$ to $S_i$. We show this observation in Lemma \ref{lma:add-vertex}. Before proving Lemma \ref{lma:add-vertex}, we state auxiliary lemma, which proves any strict run of the algorithm divides $N^+(v_0^a)$ into a set $R_i$ of light vertices and a set $H_i$ of heavy vertices with high probability. This lemma shows that the algorithm selects light vertex $x_i$ in each strict run of the algorithm. 
    
    \begin{lemma}
        \label{lma:strict-decision}
        If the strict run occurs at the $i$-th iteration, $R_i \subseteq L_{\delta / 2} (N^+(S_{i-1}^a))$ and $H_i \subseteq H_{\delta / 8} (N^+(S_{i-1}^a))$ hold with probability at least $1 - O(1/n^7)$.
    \end{lemma}
    
    \begin{proof}
        Since $S_i^a$ is nonempty and its cardinality is monotonically increasing, we have $|S_i^a| \ge 1$, and thus $\Gamma = N^+(S_i^a) \geq \delta$ holds at 
        the beginning of the strict run at the $i$-th iteration. This implies $|\Gamma| \ge \alpha = \delta/8$. By Corollary~\ref{crl:detect-heavy}, $R_i \subseteq L_{\delta / 2}(N^+(S_{i-1}^a))$ and 
        $H_i \subseteq H_{\delta / 8} (N^+(S_{i-1}^a)) $ holds with probability at least $1 - 1/n^7$.
    \end{proof}
    
    \begin{lemma}
        \label{lma:add-vertex}
        For any $i$, $x_{i}$ is $\delta / 2$-light for $N^+(S_i^a)$
        with probability at least $1 - O(1/n^7)$.
    \end{lemma}
    
    \begin{proof}
        We first consider the case that $x_{i}$ is added without strict runs. 
        In this case, agent $a$ directly visits $x_{i}$ and checks its heaviness. Hence, the lemma obviously holds. We next consider the case that $x_{i}$ 
        is added after the strict run. By Lemma \ref{lma:strict-decision},  $R_{i+1} \subseteq L_{\delta / 2}(N^+(S_{i}^a))$ holds with probability at least $ 1 - 1/n^7$. Thus any vertex $v \in R_{i+1}$ is $\delta/2$-light for $N^+(S_i^a)$. Hence, the lemma holds. 
    \end{proof}
    
    Now we show that in each iteration $H_{i+1} \subseteq H_{\delta/8}(N^+(S_{i}^a))$ holds.  
    
    \begin{lemma}
        \label{lma:density}
        For any $i \in [1, n-1]$, let $Y_i$ be the indicator random variable taking
        $Y_i = 1$ if and only if $H_{i+1} \subseteq H_{\delta / 8} (N^+(S_{i}^a))$ holds. Then we have $\Pr \left[\bigcap_{i=1}^{n} Y_i = 1 \right] \ge 1 - O(1/n^6)$.
    \end{lemma}
    
    \begin{proof}    
        Since $S_i^a$ is nonempty and its cardinality is monotonically increasing, we have $|S_i^a| \ge 1$, and thus $\Gamma = N^+(S_i^a) \geq \delta$ holds at the 
        beginning of the strict run in the $i$-th iteration. It implies $|\Gamma| \ge \alpha = \delta/8$. By Lemma \ref{lma:add-vertex}, $ |N^+(S_{i-1}^a) \cap N^+(x_i)| < \delta / 2 $ holds, and then we have $|N^+(S_{i-1}^a) \setminus N^+(x_i)| \geq \delta / 2  > \alpha$. Hence any call of 
        \textsf{Sample} satisfies the assumption of Corollary~\ref{crl:detect-heavy}
        with probability at least $1 - 3/n^7$. Since \textsf{Sample} is called at most 
        $O(n)$ times, a standard union-bound argument provides the lemma. 
    \end{proof}

    By using Lemma \ref{lma:density}, we prove that the algorithm eventually finds a $(a, \delta/8, 2)$-dense set $T^a$ in Lemma \ref{lma:iteration-bound}. We also prove the upper bound for the number of iterations of the algorithm. 
    
    \begin{lemma}
        \label{lma:iteration-bound}
        Algorithm \textsf{Construct} outputs a $(a, \delta/8, 2)$-dense set $T^a$ within $O(n / \delta)$ iterations 
        with probability at least $1 - O(1/n^5)$.
    \end{lemma}
    
    \begin{proof}
        Let $T^a = N^+(S_j^a)$. That is, the algorithm terminates at the $j$-th iteration.
        First we show that $T^a$ is 
        $(a, \delta/8, 2)$-dense. Since $S_i^a \subseteq N^+(v_0^a)$ holds,
        the first and second conditions of $(a, \delta/8, 2)$-dense condition
        are obviously satisfied. Consider the third condition.   
        By definition, two 
        sets $R_i$ and $H_i$ are always a partition of $N^+(v_0^a)$. 
        Thus we obtain $H_j = N^+(v_0^a)$ because $R_j = \emptyset$ holds. 
        Lemma~\ref{lma:density} implies that $N^+(v_0^a) = H_j \subseteq 
        H_{\delta/8}(N^+(S_j^a))$ holds. That is, $T^a = N^+(S_j^a)$ satisfies
        the third condition.
    
        We next show that the event $R_i = \emptyset$ occurs within $O(n / \delta)$ iterations. By Lemma \ref{lma:add-vertex}, $|N^+(x_i) \setminus N^+(S_{i-1}^a)| \geq \delta / 2$ holds for any $x_i$. Then we have $|N^+(S_j^a)| \geq j \delta / 2$. Due to the trivial upper bound of $|N^+(S_j^a)| \leq n$, 
        we obtain $j \leq 2n/ \delta = O(n / \delta)$. The success probability of the lemma is derived from taking the union bound for at most $O(n)$ applications of Lemmas \ref{lma:add-vertex} and \ref{lma:density}. 
    \end{proof}

    We analyse the time complexity of the algorithm \textsf{Construct}.
    
    \begin{lemma}
        \label{time-comp}
        The total running time of \textsf{Construct} is $O(n \log^2 n / \delta)$ time with probability at least $1 - O(1/n^3)$. 
    \end{lemma}
    
    \begin{proof}
        We first bound the total running time incurred by the part of 
        optimistic decision. Assume that $T^a$ is constructed at the $j$-th iteration. 
        For each $1 \geq i \geq j-1$, the optimistic run of \textsf{Sample}$(N^+(x_i) \setminus N^+(S_{i}^a), \delta/8)$ takes $96 \lceil |(N^+(x_i) \setminus N^+(S_{i}^a)| \ln n/ \delta \rceil$ rounds. Hence, the total running time is
        bounded by
        \begin{align*}
        \sum_{i=1}^{r} 96 \left\lceil \frac{ |N^+(x_i) \setminus N^+(S^a_{i})| \ln n}{\delta} \right\rceil &\leq O\left(\frac{ N^+(S^a_j) \log n }{ \delta }\right) \\ 
        &= O \left( \frac{n \log n}{\delta} \right).   
        \end{align*}
    
        We next consider the time complexity caused by the part of strict decision. We show that \textsf{Sample} is executed as a strict run at most $O(\log n)$ times. It is sufficient to prove that 
        at least a constant fraction of $R_i$ is moved to $H_{i+1}$ with high probability if the strict run occurs at the $i$-th 
        iteration.
        In each $i$-th iteration, let $g_i$ be the number of 
        $(\delta / 8)$-heavy vertices for $N^+(S_i^a)$. We show that $g_i / |R_i| \ge 1/2$ holds if the agent samples no light vertex from $R_i$ in the strict run of \textsf{Sample}. Consider the case of $g_i / |R_i| \leq 1/2$. Then the probability that the agent samples a $\delta/8$-heavy vertex is at most $1/2$. Thus, the probability that all of the sampled vertices are $\delta/8$-heavy is at most $(1/2)^{\lceil 4 \log n \rceil} \leq 1/n^4$. Conversely, if all of the sampled vertices are $\delta/8$-heavy, $g_i / |R_i| \geq 1/2$ holds with probability at least $1 - 1/n^4$. By Lemma \ref{lma:strict-decision}, 
        the strict run of \textsf{Sample} in the $i$-th iteration moves all the $\delta/8$-heavy vertices in $R_i$ to $H_{i+1}$ with high probability. Then at least a half of the elements in $R_i$ are deleted. Since the cardinality of $R_i$ never increases, the number of calls to \textsf{Sample} as a strict run is at most 
        $O(\log n)$ times with high probability. Each strict run takes $O((|N^+(S_i^a)| \log n ) / \delta)$ rounds, and thus
        the total running
        time of \textsf{Construct} is bounded by $O(( n\log^2 n ) / \delta)$. The success probability of the lemma is obtained by taking union bounds on $O(\log n)$ applications of Lemma \ref{lma:strict-decision}. 
    \end{proof} 
    
    Finally, we obtain the main lemma of \textsf{Construct}.
    
    \begin{lemma}
        Algorithm \textsf{Construct} outputs $T^a$ satisfying $(a, \delta/8, 2)$-dense condition in $O(n \log^2 n / \delta)$ rounds with probability at least $1 - O(1/n^3)$.
    \end{lemma}
    
    The combination of this lemma and Lemma~\ref{lemma:dense-rendezvous} deduces the correctness of our rendezvous algorithm.

    \begin{theorem}
        Let $G = (V, E)$ be any graph such that $\delta_G \ge \sqrt{n}$ holds. There is an algorithm that completes rendezvous within $O \left( \frac{n}{\delta} \log^2 n + \sqrt{\frac{n \Delta}{\delta}} \log n \right)$ rounds with high probability.
    \end{theorem}

    \section{Discussion}
\label{sec:discussion}

\subsection{Removing the Assumption of Min-Degree Knowledge}
\label{subsec:remove-assumption}
In the algorithm presented in Subsection \ref{sec:constructT}, we suppose that agents know a constant factor approximation of $\delta$. This assumption can be easily removed by 
a simple doubling-estimation mechanism. Precisely, in the construction of $T^a$ (which is the only  
part of the algorithm using $\delta$), agent $a$ initially sets $\delta'$ to the half of the degree of $v_0^a$. If the agent visits a vertex whose degree is less than $\delta'$, then it restarts the procedure of \emph{Construct} after halving $\delta'$. Note that we do not have to restart agent $b$ for 
synchronization because its behavior (in \emph{Main-rendezvous}) is inherently oblivious (i.e., iteratively marking neighbors). Eventually the procedure terminates without restarting
when $\delta' < \delta_G$ is satisfied. Since the running time of \textsf{Construct} is $O( ( n \log^2 n ) / \delta')$, the doubling
update of $\delta'$ does not incur any extra asymptotic cost. That is, if the estimation of $\delta'$ starts from
a range $[2^j, 2^{j+1}]$, the total running time is bound as follows:
\begin{align*}
\sum_{\lfloor \log \delta \rfloor \leq j' \leq j} &O(n\log^2 n / 2^{j'}) \\  
&= O(n\log^2 n / \delta) \cdot 
\left(1 + \frac{1}{2} + \dots + \frac{1}{2^{j - \lfloor \log \delta \rfloor}}\right) \\ 
&=O\left(\frac{n\log^2 n}{\delta}\right).
\end{align*}

\begin{corollary}
        The modified algorithm stated above outputs $T^a$ (equivalently, 
        $N^+(S_i^a)$) satisfying $(a, \delta'/8, 2)$-dense set in 
        $O(n \log^2 n / \delta')$ rounds with probability at least $1 - O(1/n^3)$.
\end{corollary}

\subsection{Algorithm without Using Whiteboards}
In this subsection, we present a rendezvous algorithm \textsf{Rendezvous-without-Whiteboard} that does not use whiteboards, under the assumption
that nodes are tightly named (that is, $n' = O(n)$). We present the pseudo-code of the algorithm in Algorithm \ref{alg:rendezvous-without-wb}. This algorithm assumes that agents know the value of $n'$ and the minimum degree $\delta$, but the minimum-degree assumption can be removed by the technique in Section~\ref{subsec:remove-assumption}. In this algorithm, agent $a$ first constructs a set $T^a \subseteq N^+(N^+(v_0^a))$ in the same way as the original one (recall that \textsf{Construct} does not use whiteboards). 
In order to synchronize the iterative probings of vertices by both agents, they start \textsf{Rendezvous-without-Whiteboard} at round 
$t' = c_1 n' \log^2 n / \delta$ for sufficiently large constant $c_1$ such that the construction of $T^a$ 
finishes by round $t'$.

We define several notations. We denote the ID space $\{ 1, \dots, n' \}$ by 
$\mathbb{S}_{\mathit{ID}}$. For any integer $\beta$, we define the $\beta$-partition 
$\{ \mathbb{I}_{1} \dots, \mathbb{I}_{\lceil n/\beta \rceil} \}$ of $\mathbb{S}_{\mathit{ID}}$ as $\mathbb{I}_i = [(i-1)\beta + 1, i\beta] \}$ for all $i$. 
The goal of the algorithm is that for an appropriate $\beta$, the agents $a$ and 
$b$ respectively construct $\Phi_{a} \subseteq T^a$ and $\Phi_{b} \subseteq N^+(v_0^b)$ satisfying the following properties with high probability:
\begin{itemize}
   \item (intersection) $|\Phi_a \cap \Phi_b| \ge 1$.
   \item (sparseness) There exists some constant $c_2$ such that $| \Phi_a \cap \mathbb{I}_i | \le c_2 \log n$ and $|\Phi_b \cap \mathbb{I}_i| \le c_2 \log n$ hold for any $i \in [1, \lceil n/\beta \rceil]$. 
\end{itemize}
We first present the construction of $\Phi_a$ and $\Phi_b$ satisfying the properties above. For each $v \in T^a$, agent $a$ adds $v$ into $\Phi_a$ with probability $4 \ln n / \sqrt{\delta}$. Similarly, for each $v \in N^+(v_0^b)$, agent $b$ adds $v$ into $\Phi_b$ with probability $4 \ln n / \sqrt{\delta}$. Then we can guarantee with high probability that $\Phi_a$ and $\Phi_b$ satisfy the intersection property,
and also satisfy the sparseness property for 
$\beta = \lceil \sqrt{\delta} \rceil$ and $c_2 = 18$. 

We explain how rendezvous is achieved by using two sets $\Phi_a$ and $\Phi_b$. The agents $a$ and $b$ iterate the following operations for all 
$i = 1, 2, \dots, \lceil n / \sqrt{\delta} \rceil$ (referred as $i$-th \emph{phase} of agents $a$ and $b$). 
The $i$-th phase consists of $ \lceil 4 c_2 \ln n \rceil^2$ rounds, and starts at round $t' + (i-1)  \lceil 4 c_2 \ln n \rceil^2 + 1$. In the $i$-th phase, agent $a$ visits each vertex 
$v_j \in \Phi_a \cap \mathbb{I}_i$ in ascending order of its ID, and waits $ \lceil 4c_2 \ln n \rceil$ rounds 
at each visited vertex. After visiting all the vertices in $\Phi_a \cap \mathbb{I}_i$, the agent waits at the initial position until round $t' + i \lceil 4 c_2 \ln n \rceil^2$ to synchronize the next phase. The behavior of agent $b$ is similar to that of $a$. It visits each $v_k \in \Phi_b \cap \mathbb{I}_i$ in ascending order of its ID. The agent $b$ waits
at each visited vertex for two rounds. Agent $b$ repeats this process $\lceil 4 c_2 \ln n \rceil$ times. Then it 
waits on the initial position until $t' + i \cdot \lceil 4 c_2 \ln n \rceil^2$ rounds. We can show that agents $a$ and $b$ attain
 rendezvous in $\mathbb{I}_l$ such that $\Phi_a \cap \Phi_b \cap \mathbb{I}_l \neq \emptyset$ holds. The total time complexity is $O((n / \beta) \cdot \log^2 n) = O((n \log^2 n) / \sqrt{\delta})$ rounds.

\begin{algorithm}[htp]
   \caption{\textsf{Rendezvous-without-Whiteboards}}
   \label{alg:rendezvous-without-wb}
   \textbf{Operations of Agent $A$}
   
   \begin{algorithmic}[1]
    \setlength{\algorithmicindent}{15pt}
    \STATE \textsf{Construct}
    \STATE wait until $t = c_1 \left(\frac{n' \log^2 n}{\delta}\right)$
    \FORALL{$u \in T^a$}
        \STATE $\Phi^a \leftarrow \Phi^a \cup \{ u \}$ with probability $\frac{4 \log n}{\sqrt{\delta}}$
    \ENDFOR
    \FOR{$ i = 1 ~ \text{to} ~  \lceil n / \sqrt{\delta} \rceil $  }
        \FORALL{$u \in \Phi^a \cap \mathbb{I}_i $}
            \STATE visit $u$
            \STATE wait on $u$ until $\lceil 4 c_2 \log n \rceil$ time (including the round moving to $u$)
            \STATE return to $v^a_0$
        \ENDFOR
        \STATE wait on $v^a_0$ until time $c_1 \left(\frac{n' \log^2 n}{\delta}\right) + i\lceil 4 c_2 \log n \rceil^2$
    \ENDFOR 
\end{algorithmic}

   \textbf{Operations of Agent $B$}
   \begin{algorithmic}[1]
    \setlength{\algorithmicindent}{15pt}
    \FORALL{$u \in N^+(v_0^b)$}
        \STATE $\Phi^b \leftarrow \Phi^b \cup \{ u \}$ with probability $\frac{4 \log n}{\sqrt{\delta}}$
    \ENDFOR
    \STATE wait until $t = c_1 \left(\frac{n' \log^2 n}{\delta}\right)$
    \FOR{$i = 1$ to $ \lceil n / \sqrt{\delta} \rceil$}
        \FOR{$j = 1$ to $\lceil 4 c_2 \log n \rceil$}
            \FORALL{$u \in \Phi^b \cap \mathbb{I}_i$}
                \STATE visit $u$
                \STATE wait two time units on $v^b_0$
                \STATE return to $v^b_0$
            \ENDFOR
        \ENDFOR
        \STATE wait on $v^b_0$ until $t = c_1 \left(\frac{n' \log^2 n}{\delta}\right) + i\lceil 4 c_2 \log n \rceil^2$
    \ENDFOR
\end{algorithmic}
\end{algorithm}
   
\begin{theorem}
   Algorithm \textsf{Rendezvous-without-Whiteboard} achieves rendezvous in $O \left( t' + \frac{n}{\sqrt{\delta}} \log^2 n \right)$ rounds with probability at least $1 - O(1/n^2)$.
\end{theorem}

\begin{proof}
First, we show that $\Phi_a$ and $\Phi_b$ satisfy the intersection property. 
By the independence of the probabilistic choices of agents $a$ and $b$, 
any node in $T^a \cap N^+(v_0^b)$ is contained in both $\Phi_a$ and $\Phi_b$ 
with probability $(4 \ln n / \sqrt{\delta})^2 = (4 \ln n)^2 /\delta$. Hence the probability $p$ that $|\Phi_a \cap \Phi_b| = 0$ is upper bounded by
$ p \le \left( 1 - \frac{ (4 \ln n)^2 }{\delta} \right)^{\delta / 8} \le \mathrm{e}^{2 \ln^2 n} \le \frac{1}{n^2}. $
That is, the intersection property is satisfied with high probability.
Next, we show that $\Phi_a$ and $\Phi_b$ satisfy the sparseness property. 
For any $i \in [1, \lceil n /  \sqrt{\delta}  \rceil] $, let 
$Y^a_i$ be the number of vertices in $N^+(v_0^a) \cap \mathbb{I}_i$. Then we have $\expect{Y^a_i} \le \lceil \sqrt{\delta} \rceil \cdot 4 \ln n / \sqrt{\delta} \le 9 \ln n$. Applying the Chernoff bound, the probability $\Pr[Y^a_i \ge 18 \log n]$ is upper bounded by 
$
   \Pr[Y^a_i \ge 18 \ln n] \le \Pr[Y^a_i \ge (1 + 1) 9 \ln n] \le e^{3 \ln n} \le \frac{1}{n^3}. 
$ By taking union bound over all $i \in [1, \lceil n /  \sqrt{\delta} \rceil]$, $a$, and $b$, the probability that $\Phi_a$ and 
$\Phi_b$ do not satisfy the sparseness property is at most $3/n^2$. 

Finally, we show that if $\Phi_a$ and $\Phi_b$ satisfy the two properties, then rendezvous is achieved within $O((n \log^2 n) / \sqrt{\delta})$ rounds. 
We consider the $l$-th part such that $|\mathbb{I}_l \cap \Phi_a \cap \Phi_b| \ge 1$ holds. Let $r$ be any vertex in $|\mathbb{I}_l \cap \Phi_a \cap \Phi_b|$, and $s$ be the order of $r$ in $\Phi_a \cap \mathbb{I}_l$ following IDs. 
By the definition of the algorithm, both $a$ and $b$ starts phase $l$ at round $t' + (l-1) \lceil 4 c_2 \ln n \rceil^2 + 1$. In addition, the time when agent $a$ stays at $r$ is from round $t' + (i-1) \lceil 4 c_2 \ln n \rceil^2 + (s-1)\lceil 4 c_2 \ln n \rceil -2$ to $t' + (i-1) \lceil 4 c_2 \ln n \rceil^2 + s\lceil 4 c_2 \ln n \rceil -2$.
During that period, agent $b$ visits all the nodes in $\Phi_b \cap \mathbb{I}_l$.
That is, rendezvous is achieved. 
\end{proof}

\section{Impossibility for Sub-linear Time Rendezvous}
\label{sec:impossibility}

In this section, we show four impossibility results for sublinear-time rendezvous, which respectively concern the four unconventional assumptions of our algorithm, namely, bounded minimum degrees,  accessibility to neighborhood IDs, initial distance one, and randomization. In each proof, we show the impossibility results in the models relaxing the corresponding assumption. We define some terminologies used in the proofs. Given a graph $G$ and an algorithm $\mathcal{A}$, let $\hat{X}(G, a, v, f(n))$ be the random variable representing the set of vertices visited by agent $a$ initially at vertex $v$ in $G$ in the first consecutive $f(n)$ rounds. While this is an illegal run because $b$ is not in the graph, but can identify the (probabilistic) set of vertices $a$ visits. Also, we define $X(G, a, v, f(n))$ to be the vertex set defined as $X(G, a, v, f(n)) = \{ x \in V(G) ~|~ \Pr [x \in \hat{X} (G, a, v, f(n))] \le 1/4 \}$. 

\subsection{Lower bound in the Case of bounded minimum degrees}

First, we show that there is a graph instance with minimum degree $\delta = o(\sqrt{n})$ and $\Delta = \omega(\sqrt{n})$ 
such that any algorithm needs $\Omega(\Delta)$ rounds for neighborhood rendezvous. Precisely, the $\Omega(n  / \delta)$-round 
lower bound is obtained in the graphs with $\delta = o(\sqrt{n})$ and $\Delta = \Omega(\sqrt{n})$.

\begin{theorem}
    \label{thrm:smalldelta}
    Letting $\delta = o(\sqrt{n})$ and $\Delta = \omega(\sqrt{n})$, the $(\Delta, \delta, 1)$-rendezvous problem has
    a class of instances where any rendezvous algorithm takes $\Omega(\Delta)$ rounds with a constant probability. 
    In particular, the $(n/2, 1, 1)$-rendezvous problem has a class of instances where any rendezvous algorithm 
    takes $\Omega(n)$ rounds with a constant probability.
\end{theorem}

\begin{proof}
    We first consider the case of $\Delta = n/2$ and $\delta = 1$ for simplicity of argument.
    Suppose for contradiction that an algorithm $\mathcal{A}$ achieves rendezvous within $f(n) = o(n)$ rounds 
    with high probability for the
    $(n/2, 1, 1)$-rendezvous problem. Assume that $n$ is a multiple of $4$ for simplicity, and let $[1, n]$ be
    the domain of vertex IDs. First, we consider a star graph $S_1(j)$ of $n/2 + 1$ vertices, where the ID of the center is $j \in [n/2 + 1, n]$, and IDs of all leaves are from $[1, n/2]$. In this graph we put agent $a$ at the center vertex $j$, and run $\mathcal{A}$ during $f(n)$ rounds. It is easy to verify $|X(S_1(j), a, j, f(n))| > n/4$ because $f(n)$ is sublinear of $n$. Next, we consider a star graph $S_2(k)$ of $n/2 + 1$ vertices that 
    consists of the center vertex with ID $k \in [1, n/2]$ and leaf sets with IDs $[n/2+1, n]$. It also satisfies $|X(S_2(k), b, k, f(n))| > n/4$.
    Now we consider a directed bipartite graph $G' = ([1, n/2], [n/2 + 1, n], E)$. The edge set $E$ is defined as
    $E = \{ (h, i) \mid h \in X(S_1(i), a, i, f(n)) \vee h \in X(S_2(i), b, i, f(n)) \}$. Since we have $|X(S_1(i), a, i, f(n))| > n/4$ and $|X(S_2(i), b, i, f(n))| > n/4$ for all $i$, the total number of directed 
    edges is more than $(n/2 \cdot n/4) \cdot 2 = n^2 / 4$. This means that there exists at least one pair $(j, k)$ such that both $(j, k)$ and $(k, j)$ are contained in $E$.
    We consider the graph that consists of two star graphs of $n/2 + 1$ vertices
    sharing an edge (Fig. \ref{fig:smalldelta} (a)). The IDs of the two center vertices are $j$ and $k$, and the IDs of $j$'s leaves
    are from $[n/2 + 1, n] \setminus \{ k \}$, and those of $k$'s leaves are from $[1, n/2] \setminus \{ j \}$. The edge $(j, k)$ connects the two centers. In this graph, when we execute the algorithm $\mathcal{A}$ locating the two agents at $j$ and $k$ respectively, it is guaranteed that each agent does not pass through 
    edge $(j, k)$ in the first consecutive $f(n)$ rounds with probability at least $1/4$. That is, 
    the algorithm does not achieve rendezvous within $f(n)$ rounds with probability at least $1/2$. This is a contradiction. 
    
    The general case can be proven in the same way as the argument above. The only difference is to change the degree of the center vertex to $\Delta$ and replace all the leaves of star graphs with a clique of size $s = \frac{n-2}{2\Delta} =  \Omega(n / \Delta) = \Omega(\delta)$ where exactly one vertex is adjacent to the center (Fig. \ref{fig:smalldelta} (b)). That graph obviously satisfies 
    the constraint of min/max degrees, and the proof above also applies to it.
\end{proof}

\begin{figure}[h]
    \centering
    \includegraphics[width=8cm]{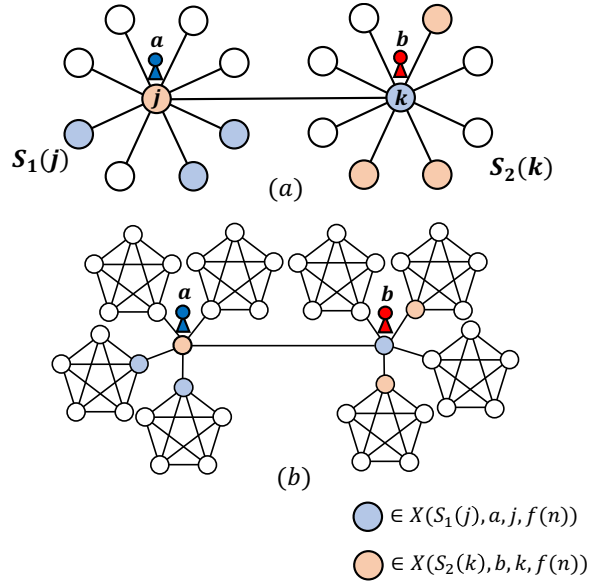}
    \caption{Proof of Theorem \ref{thrm:smalldelta}}
    \label{fig:smalldelta}
  \end{figure}

\subsection{Lower bound in the Case of the No Accessibility to IDs of Neighborhood Vertices}

Next, we show that any algorithm solving the $(\Theta(n), \Theta(n), 1)$-rendezvous problem requires $\Omega(n)$ rounds
in the worst case if agents have no access to IDs of neighborhood vertices. 
\begin{theorem}
    \label{thrm:noaccess}
    Let $n$ be even, $n \ge 6$, $\delta = n/2-1$ and $\Delta = n/2-1$, and assume that any agent has no access to neighborhood IDs. 
    Then there exists an instance of $(\Delta, \delta, 1)$-rendezvous problem where any rendezvous algorithm 
    takes $\Omega(\Delta)$ rounds with a constant probability. 
\end{theorem}

\begin{proof}
     Suppose for contradiction that an algorithm $\mathcal{A}$ achieves rendezvous within $f(n)= o(n)$ rounds with high probability for the $(n/2-1, n/2-1, 1)$-rendezvous problem. We first consider two cliques $C_1$ and $C_2$ of $n/2$ vertices where each vertex has an arbitrary ID. Let agent $a$ be located at $v_{0}^a$ in the clique $C_1$, and let agent $b$ be located at $v_{0}^b$ in the clique $C_2$. As the proof of Theorem \ref{thrm:smalldelta}, we make agents $a$ and $b$ execute algorithm $\mathcal{A}$ in each clique. By the assumption of $f(n) = o(n)$, it is easy to verify that
    $|X(C_1, a, v_0^a, f(n))| > n/4$ and $|X(C_1, b, v_0^a, f(n))| > n/4$ holds. Now we select vertices $x_1 \in X(C_1, a, v_0^a, f(n))$ and $x_2 \in X(C_2, b, v_0^b, f(n))$. Let $j = \hat{P}^
    {-1}_{v_0^a}(x_1)$, $k = \hat{P}^
    {-1}_{v_0^b}(x_2)$, $\bar{j} = \hat{P}^
    {-1}_{x_1}(v_0^a)$, 
    and $\bar{k} = \hat{P}^
    {-1}_{x_2}(v_0^b)$. We construct a graph $G$ by removing edges $(v_0^a, x_1)$ and $(v_0^b, x_2)$ from $C_1$ and $C_2$ respectively, and adding the edges $(v_0^a, v_0^b)$ and $(x_1, x_2)$. The local 
    port number of those edges are defined as $\hat{P}^
    {-1}_{v_0^a}(v_0^b) = j$, $\hat{P}^
    {-1}_{v_0^b}(v_0^a) = k$,  
    $\hat{P}^
    {-1}_{x_1}(x_2) = \bar{j}$, and $\hat{P}^
    {-1}_{x_2}(x_1) = \bar{k}$. The construction is illustrated in Fig. \ref{fig:noaccess}. Consider the $f(n)$-round run of $\mathcal{A}$ in $G$ 
    where two agents $a$ and $b$ start from $v_0^a$ and $v_0^b$ respectively. Since $v_0^a$ and $v_0^b$ are connected by an edge, this is an instance of the $(n/2-1, n/2-1, 1)$-rendezvous problem. Since $x_1 \in X(C_1, a, v_0^a, f(n))$, agent 
    $a$ visits $x_1$ or $v_0^b$ with probability at most $1/4$. Similarly, $b$ also visits $x_2$ or $v_0^a$ with
    probability at most $1/4$. This implies that with probability at least $1/2$ no agent moves along edge
    $(v_0^a, v_0^b)$ or $(x_1, x_2)$, that is, rendezvous is not achieved at round $n/2$ with a constant
    probability. This is a contradiction.
\end{proof}

\begin{figure}[h]
    \centering
    \includegraphics[width=8cm]{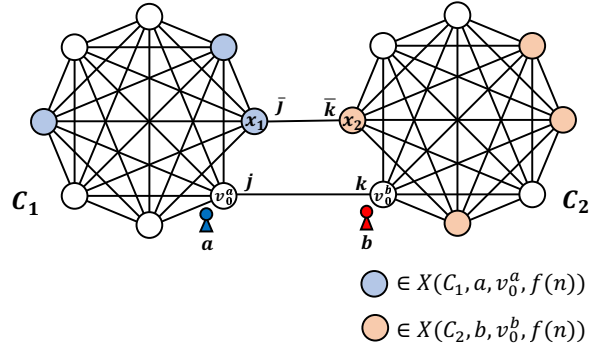}
    \caption{Proof of Theorem \ref{thrm:noaccess}}
    \label{fig:noaccess}
  \end{figure}

  \subsection{Lower bound in the Case of the distance two of initial locations}

Next, we show that the lower bound for the $(\Theta(n), \Theta(n), 2)$-rendezvous problem.

\begin{theorem}
    \label{thrm:distance2}
    Let $n$ be odd, $\Delta = n-1$ and $\delta = (n-1)/2$. $(\Delta, \delta, 2)$-rendezvous problem has a graph instance where any algorithm takes $\Omega(\Delta)$ rounds with a constant probability.
\end{theorem}

\begin{proof}
    Suppose for contradiction that an algorithm $\mathcal{A}$ achieves rendezvous within $f(n) = o(n)$ rounds with high probability for the $(n-1, (n-1)/2, 2)$-rendezvous problem. We first consider two cliques $C_1, C_2, \dots, C_{(n+1)/2}$ of $(n+1)/2$ vertices, where the $i$-th vertex set is $V(C_i)$. The IDs of the vertices of each clique $C_i$ are assigned from $\left[ \frac{n+1}{2}(i-1) + 1, \frac{n+1}{2} i \right]$ respectively for all $i \in [1, \frac{n+1}{2}]$. 
     Suppose that agent $a$ executes 
    algorithm $\mathcal{A}$ in each clique $C_i$ with an arbitrary initial location $c_i \in V(C_i)$. By the assumption of $f(n) = o(n)$, it is easy to verify that 
    $|X(C_i, a, c_i, f(n))| > (n+1)/4$. Let $V'$ a vertex set that obtained by picking up one vertex $w_i \in |X(C_i, a, c_i, f(n))|$ for all $i \in [1, \frac{n+1}{2}]$, and we construct a clique $C'$ 
    consisting of $(n+1)/2$ vertices whose IDs come from  $V'$. Suppose that agent $b$ executes algorithm $\mathcal{A}$ in $C'$ with an arbitrary initial location $c' \in V(C')$. It also satisfies $|X(C', b, c', f(n))| > (n+1)/4$ because $f(n)$ is sublinear of $n$. We pick up any vertex 
    $x \in X(C', b, c', f(n))$. Letting $C_k$ be the clique containing the vertex $x$, we construct the graph $G$ consisting of two cliques $C'$ and $C_k$ sharing $x$ (Fig. \ref{fig:distance2}). 
    Consider the $f(n)$-round run of $\mathcal{A}$ in $G$ where $a$ and $b$ respectively start from
    $c_k$ and $c'$. This is an instance of $(n - 1, (n-1)/2, 2)$-rendezvous problem. 
    Since $x \in X(C_k, a, c_k, f(n)) \cap X(C', b, c', f(n))$ holds,  $a$ and $b$ do not visit $x$ with probability at least $1/4$. That is,  
    they cannot attain the rendezvous within $f(n)$ rounds at least with probability $1/2$. This is a contradiction. 
\end{proof}

\begin{figure}[h]
    \centering
    \includegraphics[width=8cm]{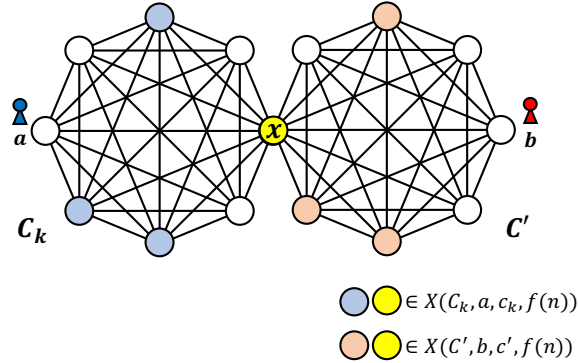}
    \caption{Proof of Theorem \ref{thrm:distance2}}
    \label{fig:distance2}
  \end{figure}

  \subsection{Lower bound for deterministic algorithms}
  \label{subsec:deterministiclowerbound}

We show that any deterministic algorithm solving the $(\Theta(n), \Theta(n), 1)$-rendezvous problem requires $\Omega(n)$ rounds in the worst case. First, we outline the proof strategy. Suppose for contradiction that an algorithm 
$\mathcal{A}$ solves $(\Theta(n), \Theta(n), 1)$-rendezvous problem within $o(n)$ rounds. In the proof, 
we adaptively construct the hard-core instance according to the behavior of $\mathcal{A}$: We start the construction with the two star graphs whose centers are the initial locations of two agents, and consider the run of $\mathcal{A}$ in that graph. When the agent moves to an unvisited vertex, we adaptively fix its neighborhood vertices. More precisely, the graph construction roughly follows the process below: We select in advance $\Omega(n)$ vertices as a pool, and if an agent moves to an unvisited vertex with degree $o(n)$, we select $\Omega(n)$ vertices from the pool as neighbors. This construction provides two independent graphs respectively associated with two agents. Finally, we carefully glue them in the way of guaranteeing the initial distance one and minimum degree $\Omega(n)$, which
becomes the instance yielding $\Omega(n)$-round lower bound.

We define some notations for explaining the details. Let $n$ be a multiple of $32$ for simplicity. As we stated, our proof first constructs two instances (for two agents) separately. By symmetry 
we only focus on the instance for agent $a$. We select 
an arbitrary ID space $\mathit{ID}_a$ whose size is $n/2 + 1$ for the instance of agent $a$, and fix an initial vertex $v^a_0 \in \mathit{ID}_a$.
Let $Q_t^a(\mathcal{A}, G, v^a_0) = \{v^a_0, v^a_1, \dots, v^a_{t}\}$. That is, 
$Q_t^a(\mathcal{A}, G, v^a_0)$ is the set of vertices visited by agent $a$ in the execution of $\mathcal{A}$ starting from $v^a_0$ in $G$ up to round $t$. We also define the sequence $S_t^a(\mathcal{A}, G, v^a_0) = (v^a_0, v^a_1, \dots, v^a_{t})$ of the vertices in 
$Q_t^a(\mathcal{A}, G, v^a_0)$ with order. 
Given $\mathcal{A}$, $G = (V, E)$, $v_0^a$ and a round $r \ge 0$, we can construct the \textit{execution spanning subgraph} $\hat{G}_r^a(\mathcal{A}, G, v_0^a) = (\hat{V}, \hat{E})$ such that $\hat{V} = N^{+}_{G}(Q_r^a(\mathcal{A}, G, v^a_0))$ and $\hat{E} = \{(u, v) ~|~ u \in Q_r^a(\mathcal{A}, G, v^a_0) \wedge (u, v) \in E \}$. Intuitively, $\hat{G}_r^a(\mathcal{A}, G, v_0^a)$ represents the 
substructure of $G$ seen by agent $a$  in the execution of $\mathcal{A}$ starting from $v_0^a$
up to round $r$. Now we assume any graph $G'$ such that $\hat{G}_r^a(\mathcal{A}, G, v_0^a) = \hat{G}_r^a(\mathcal{A}, G', v_0^a)$ holds. It is obvious that the behavior of $a$ in $G'$ starting from $v^a_0$ is 
completely same as that in $G$ up to round $r+1$, and thus we obtain the following proposition.
\begin{proposition}
    \label{prop:construction}
    Assume for any $G, G'$, we have $\hat{G}_r^a(\mathcal{A}, G, v_0^a)$ $ = \hat{G}_r^a(\mathcal{A}, G', v_0^a)$. Then, $S_{r+1}^a(\mathcal{A}, G, v_0^a)= S_{r+1}^a(\mathcal{A}, G', v_0^a)$ holds.
\end{proposition}  

We show the lemma below, which is a key observation of our lower bound proof.

\begin{lemma}
    \label{lma:constructgraph}
    Let $\mathcal{A}$ be any algorithm terminating within $t \le n / 32$ rounds. 
    Suppose that $ID_a$ and $v_0^a$ is given. There exists a graph $G$ containing a vertex subset 
    $W \subseteq N_{G}(v_0^a)$ of size at least $13n/32$ such that (i) $(Q_t^a(\mathcal{A}, 
    G, v^a_0) \setminus \{v^a_0\}) \cap N^{+}_{G}(W) = \emptyset$ holds, and 
    (ii) for each vertex $w \in V(G) \setminus ( N^{+}_{G}(W) \setminus \{ v_0^a \})$, $|N_G(w)| = \Theta(n)$ holds.  
\end{lemma}

\begin{proof}
   We adaptively construct the graph $G$ according to the agent $a$'s movement. Precisely, we incrementally fix the sequence of graphs $G_0, G_1, \dots, G_{t}$ such that  for each $r \in [0, t-1]$, $S_{r+1}^a(\mathcal{A}, G_{r}, v_0^a) = S_{r+1}^a(\mathcal{A}, G_{r+1}, v_0^a)$ is guaranteed. 
   The vertex set of each $G_i$ is common, which is denoted by $V$, and equal to $ID_a$ (i.e., $V = ID_a$). Let $P \subseteq V \setminus \{v^a_0\}$ be an arbitrary subset of size $7n/16$, and $\overline{P} = V \setminus P$. We also define $E_0 = \{(v^a_0, u) \mid u \in ID_a \setminus \{v^a_0\}\} \cup \{(u, v) \mid u, v \in \overline{P} \wedge u \neq v\}$. 
   For all $r \ge 0$, the algorithm $\mathcal{A}$ outputs the vertex $v_{r+1}^a \in N_{G_r}(v_{r}^a)$, as the destination of the movement at round $r$. Let 
   $Q_r = Q_{r}^a(\mathcal{A}, G_{r}, v_0^a)$ for short. There are following two cases:
   \begin{itemize}
       \item $v_{r+1}^a \in Q_{r} \cup \overline{P}$.
       \item $v_{r+1}^a \not\in Q_{r} \cup \overline{P}$ (that is, $v_{r+1}^a \in P \setminus Q_{r}$).
   \end{itemize}
   If $v_{r+1}^a \in Q_{r} \cup \overline{P}$ holds, we simply fix $G_{r+1} = G_{r}$ (i.e., $E_{r+1} = E_{r}$). Otherwise, we construct $E_{r+1}$ by adding to $E_{r}$ the edges from 
   $v_{r+1}^a$ to all the vertices in $\overline{P} \setminus Q_{r}$.
   In the following argument, we show  $S_{r+1}^a(\mathcal{A}, G_{r}, v_0^a) = S_{r+1}^a(\mathcal{A}, G_{r+1}, v_0^a)$ holds for any $r \in [0, t-1]$ by the induction on $r$. In the base case of $r=0$, we have $Q_0 = \{ v_0^a\}$ and $S_0^a(\mathcal{A}, G_0, v_0^a) = (v_0^a)$. The algorithm outputs the vertex $v_1^a$ as the destination of the movement in $G_0$ at round $r = 0$. In any case of updating rules, we can confirm that $\hat{G}_0^a(\mathcal{A}, G_0, v_0^a) = \hat{G}_0^a(\mathcal{A}, G_1, v_0^a)$. Therefore the vertex $v_1^a$ in $G_0$ coincides with the one in $G_1$ and we have $S_{1}^a(\mathcal{A}, G_0, v_0^a)= S_{1}^a(\mathcal{A}, G_1, v_0^a)$.  
   In the case of $r > 1$, assume that we are given $G_{r}$. The algorithm outputs the vertex $v_{r+1}^a$ as the destination of the movement in $G_{r}$ at round $r$. If $v_{r+1}^a \in Q_r \cup \overline{P}$, then $G_{r} = G_{r+1}$ holds, we have $S_{r+1}(\mathcal{A}, G_{r}, v_0^a) = S_{r+1}(\mathcal{A}, G_{r+1}, v_0^a)$. Otherwise, since we add edges between unvisited vertices (from $v_{r+1}^a \in P \setminus Q_r$ to each $u \in \overline{P} \setminus Q_r$), it follows $\hat{G}_r^a(\mathcal{A}, G_{r}, v_0^a) = \hat{G}_r^a(\mathcal{A}, G_{r+1}, v_0^a)$. Then by proposition \ref{prop:construction},  $S_{r+1}^a(\mathcal{A}, G_{r}, v_0^a)= S_{r+1}^a(\mathcal{A}, G_{r+1}, v_0^a)$ hold. 

   We set $P \setminus Q_{t} = W$ (that is, the vertices in $P$ not visited by round $t$). Finally, we show that $G_t$ has the desired property of the lemma. Since the agent visits to the vertices in $P$ at most $t = n/32$ times, the size of $W$ is at least $7n/16 - n/32 = 13n/32$. Since $W$ is the set of vertices which are unvisited by agent $a$ in the execution of $\mathcal{A}$ in $G_t$, by the updating rules of the graphs, each vertex in $W$ is only connected to $v_0^a$. Therefore we have $(Q_t^a(\mathcal{A}, G, v^a_0) \setminus \{v^a_0\}) \cap N^{+}_{G}(W) = \emptyset$. Since $\overline{P}$ is a clique in $G_0$ (and thus in $G_t$), for each vertex $u \in \overline{P}$, we have $|N_{G_t}(u)| \ge n/16 -1 = \Theta(n)$. For each vertex $u \in P \cap Q_r$, the size of $\overline{P} \setminus Q_r$ is at least $n/16 - n/32 = n/32$ at any round $r \in [0, t]$, and thus we have $|N_{G_t}(u)| \ge n/32 = \Theta(n)$. 
\end{proof}

By the proposition and the lemma, we can construct the hard-core instance for the deterministic algorithm. In the proof, we apply Lemma \ref{lma:constructgraph} several times according to the agent IDs and initial positions $v_0^a, v_0^b$. Therefore in the proof we add subscripts of agent IDs and initial vertices to $G$ and $W$ constructed by the lemma, as $G_{(a, v_0^a)}$ and $W_{(a, v_0^a)}$.

\begin{theorem}
    For $\Delta = \Theta(n)$ and $\delta = \Theta(n)$, the $(\Delta, \delta, 1)$-rendezvous problem has a graph instance where any deterministic algorithm takes $\Omega(\Delta)$ rounds with probability one.
\end{theorem}

\begin{proof}
    Suppose for contradiction that a deterministic algorithm $\mathcal{A}$ achieves rendezvous within $f(n) = n/32$ rounds for the $(\Delta, \delta, 1)$-rendezvous problem of $\Delta = \Theta(n)$ and $\delta = \Theta(n)$. Let $[1, n]$ be the domain of vertex IDs. 

    We select $[1, n/2]$ and $ j \in [n/2+1, n]$ as the ID space of the execution of the agent $a$, denoted by $ID_a$. We choose $v_0^a = j$ as the initial vertex of $a$, and construct $G_{(a, j)}$ by using Lemma \ref{lma:constructgraph}. Similarly, we adaptively construct the graph instance according to the agent $b$'s moves alone. We select $[n/2+1, n]$ and $ k \in [1, n/2]$ as the ID space, denoted by $ID_b$. We choose $v_0^b = k$ as the initial vertex of $b$, and construct $G_{(b, k)}$ by also using Lemma \ref{lma:constructgraph}.

    Now we consider a directed bipartite graph $G' = ([1, n/2], [n/2 + 1, n], E)$. The edge set $E$ is defined as
    $E = \{ (x, y) \mid (x = j \wedge y \in W_{(a, j)}) \vee (x = k \wedge y \in W_{(b, k)})) \}$ for all $j$ and $k$. Since we have $|W_{(a, j)}| \ge (13/32)n > n/4$ and $|W_{(b, k)}| \ge (13/32)n > n/4$ for all $j$ and $k$, the total number of directed 
    edges is more than $(n/2 \cdot n/4) \cdot 2 = n^2 / 4$. This means that there exists at least one pair $(j, k)$ such that both $(j, k)$ and $(k, j)$ are contained in $E$.
    Finally we construct the whole graph instance. Prepare $G_{a, j}$ and $G_{b_k}$ as the subgraphs of the constructed instance. Then we add an edge between $j$ and $k$. We augment edges between any vertices in $W_{(a, j)} \setminus \{ k \}$ and in $W_(b, k) \setminus \{ j \}$ respectively. By the condition (ii) of Lemma~\ref{lma:constructgraph}, it is easy to verify that the minimum degree of the constructed instance is $\Theta(n)$. In this graph, consider the execution of $\mathcal{A}$ where two agents $a$ and $b$ are respectively located at $j$ and $k$. By the condition (i) of Lemma~\ref{lma:constructgraph}, it is guaranteed that each agent does not pass through edge $(j, k)$ in the first consecutive $n/32$ rounds. That is, the algorithm does not achieve rendezvous within $f(n)$ rounds. This is a contradiction.
\end{proof}

\section{Conclusion}
In this paper, we consider the neighborhood rendezvous problem, and propose two randomized algorithms for solving it. The first algorithm achieves rendezvous in $O \left( \frac{n}{\delta} \log^3 n + \sqrt{\frac{n \Delta}{\delta}}\log n \right)$ rounds with high probability for graphs of minimum degree $\delta = \omega(\sqrt{n} \log n)$. 
The second algorithm achieves rendezvous in $O\left( \frac{n}{\delta} \log^2 n + \frac{n}{\sqrt{\delta}}\log^2 n \right)$ rounds with high probability. It does not use whiteboards. We also presented four impossibility results for sub-linear time rendezvous, where each result respectively considers four unconventional assumptions of our algorithm, that is, bounded minimum degrees, accessibility to neighborhood IDs, initial distance one, and randomization. One can obtain the $\Omega (n)$-round lower bound if either of them is removed. Therefore we conclude that our algorithms run under a minimal assumption.

\bibliographystyle{plain}
\bibliography{refer}

\end{document}